\documentclass[11pt]{article}

\usepackage{comment}
\usepackage{algorithmicx,algorithm}
\usepackage[noend]{algpseudocode}
\usepackage{xspace}
\usepackage{fullpage}
\usepackage{amsthm}
\usepackage{amsmath,amssymb}

\newcommand{\problem}[1]{\textsc{#1}\xspace}
\newcommand{\R}{\mathbb{R}}
\newcommand{\tw}{\mathrm{tw}}
\newcommand{\td}{\mathrm{td}}

\theoremstyle{plain}
\newtheorem{theorem}{Theorem}
\newtheorem{lemma}{Lemma}
\newtheorem{claim}{Claim}
\newtheorem{corollary}{Corollary}

\title{On the Power of Tree-Depth for Fully Polynomial FPT Algorithms}

\author{Yoichi Iwata
\footnote{National Institute of Informatics 2-1-2 Hitotsubashi, Chiyoda-ku, Tokyo, Japan}
\footnote{yiwata@nii.ac.jp}
\footnote{Supported by JSPS KAKENHI Grant Number JP17K12643}
\and Tomoaki Ogasawara
\footnote{The University of Tokyo 7-3-1, Hongo, Bunkyo-ku, Tokyo, Japan}
\footnote{t.ogasawara@is.s.u-tokyo.ac.jp}
\and Naoto Ohsaka
\footnotemark[4]
\footnote{ohsaka@is.s.u-tokyo.ac.jp}
}
\date{}

\begin{document}

\maketitle

\begin{abstract}
There are many classical problems in P whose time complexities have not been improved over the past decades.
Recent studies of ``Hardness in P'' have revealed that, for several of such problems, the current fastest algorithm is the best possible under some complexity assumptions.
To bypass this difficulty, Fomin~et~al. (SODA 2017) introduced the concept of \emph{fully polynomial FPT algorithms}.
For a problem with the current best time complexity $O(n^c)$, the goal is to design an algorithm running in $k^{O(1)}n^{c'}$ time for a parameter $k$ and a constant $c'<c$.

In this paper, we investigate the complexity of graph problems in P parameterized by \emph{tree-depth}, a graph parameter related to tree-width.
We show that a simple divide-and-conquer method can solve many graph problems, including
\problem{Weighted Matching}, \problem{Negative Cycle Detection}, \problem{Minimum Weight Cycle}, \problem{Replacement Paths}, and \problem{2-hop Cover},
in $O(\mathrm{td}\cdot m)$ time or $O(\mathrm{td}\cdot (m+n\log n))$ time, where $\mathrm{td}$ is the tree-depth of the input graph.
Because any graph of tree-width $\mathrm{tw}$ has tree-depth at most $(\mathrm{tw}+1)\log_2 n$, our algorithms also run in $O(\mathrm{tw}\cdot m\log n)$ time or $O(\mathrm{tw}\cdot (m+n\log n)\log n)$ time.
These results match or improve the previous best algorithms parameterized by tree-width.
Especially, we solve an open problem of fully polynomial FPT algorithm for \problem{Weighted Matching} parameterized by tree-width posed by Fomin~et~al.
\end{abstract}

\section{Introduction}
There are many classical problems in P whose time complexities have not been improved over the past decades.
For some of such problems, recent studies of ``Hardness in P'' have provided evidence of why obtaining faster algorithms is difficult.
For instance, Vassilevska Williams and Williams~\cite{WilliamsW10} and Abboud, Grandoni and Vassilevska Williams~\cite{AbboudGW15} showed that
many problems including \problem{Minimum Weight Cycle}, \problem{Replacement Paths}, and \problem{Radius} are equivalent to
\problem{All Pair Shortest Paths (APSP)} under subcubic reductions; that is, if one of them admits a subcubic-time algorithm, then all of them do.

One of the approaches to bypass this difficulty is to analyze the running time by introducing another measure, called a \emph{parameter}, in addition to the input size.
In the theory of parameterized complexity, a problem with a parameter $k$ is called \emph{fixed parameter tractable (FPT)} if it can be solved in $f(k) \cdot |I|^{O(1)}$ time for some function $f(k)$ that does not depend on the input size $|I|$.
While the main aim of this theory is to provide fine-grained analysis of NP-hard problems, it is also useful for problems in P.
For instance, a simple dynamic programming can solve \problem{Maximum Matching} in $O(3^\tw m)$ time, where $m$ is the number of edges and
$\tw$ is a famous graph parameter called \emph{tree-width} which intuitively measures how much a graph looks like a tree (see Section~\ref{sec:pre} for the definition).
Therefore, it runs in linear time for any graph of constant tree-width, which is faster than the current best $O(\sqrt{n} m)$ time for the general case~\cite{blum1990new,vazirani1994theory,DBLP:journals/jacm/GabowT91}.

When working on NP-hard problems, we can only expect superpolynomial (or usually exponential) function $f(k)$ in the running time of FPT algorithms
(unless $k$ is exponential in the input size).
On the other hand, for problems in P, it might be possible to obtain a $k^{O(1)}|I|^{O(1)}$-time FPT algorithm.
Such an algorithm is called \emph{fully polynomial FPT}.
For instance,
Fomin, Lokshtanov, Pilipczuk, Saurabh and Wrochna~\cite{Fomin17} obtained an $O(\tw^4\cdot n\log^2 n)$-time (randomized) algorithm for \problem{Maximum Matching}
and left as an open problem whether a similar running time is possible for \problem{Weighted Matching}.
In contrast to the $O(3^\tw m)$-time dynamic programming, this algorithm is faster than the current best general-case algorithm already for graphs of $\tw=O(n^{\frac{1}{8}-\epsilon})$.
In general, for a problem with the current best time complexity $O(n^c)$, the goal is to design an algorithm running in $O(k^d n^{c'})$ time for
some small constants $d$ and $c'<c$.
Such an algorithm is faster than the current best general-case algorithm already for inputs of $k=O(n^{(c-c')/d-\epsilon})$.
On the negative side, Abboud, Vassilevska Williams and Wang~\cite{AbboudWW16} showed that $\problem{Diameter}$ and $\problem{Radius}$ do not admit
$2^{o(\tw)}n^{2-\epsilon}$-time algorithms under some plausible assumptions.
In this paper, we give new or improved fully polynomial FPT algorithms for several classical graph problems.
Especially, we solve the above open problem for \problem{Weighted Matching}.

\subparagraph{Our approach.}
Before describing our results, we first give a short review of existing work on fully polynomial FPT algorithms parameterized by tree-width and explain our approach.
There are roughly three types of approaches in the literature.
The first approach is to use a polynomial-time dynamic programming on a tree-decomposition, which has been mainly used for problems related to
shortest paths~\cite{chaudhuri2000shortest,planken2012computing,DBLP:conf/edbt/AkibaSK12,DBLP:conf/sigmod/Wei10}.
The second approach is to use an $O(\tw^3\cdot n)$-time Gaussian elimination of matrices of small tree-width developed by Fomin~et~al.~\cite{Fomin17}.
The above-mentioned $O(\tw^4\cdot n\log^2 n)$-time algorithm for \problem{Maximum Matching} was obtained by this approach.
The third approach is to apply a divide-and-conquer method exploiting the existence of small \emph{balanced separators}.
This approach was first used for planar graphs by Lipton and Tarjan~\cite{LiptonT80}.
Using the existence of $O(\sqrt{n})$-size balanced separators, they obtained an $O(n^{1.5})$-time algorithm for \problem{Maximum Matching} and an $O(n^{1.5}\log n)$-time algorithm for \problem{Weighted Matching} for planar graphs.
For graphs of bounded tree-width, Akiba, Iwata and Yoshida~\cite{akiba2013fast} obtained an $O(\tw\cdot (m+n\log n)\log n)$-time algorithm for \problem{2-hop Cover}, which is a problem of constructing a distance oracle, and
Fomin~et~al.~\cite{Fomin17} obtained an $O(\tw\cdot m\log n)$-time\footnote{%
While the running time shown in~\cite{Fomin17} is $O(\tw^2\cdot n\log n)$, we can easily see that it also runs in $O(\tw\cdot m\log n)$ time.
Because $m=O(\tw\cdot n)$ holds for any graphs of tree-width $\tw$, the latter is never worse than the former.
Note that $\tw\cdot n$ in the running time of other algorithms cannot be replaced by $m$ in general;
e.g., we cannot bound the running time of the Gaussian elimination by $O(\tw^2\cdot m)$, where $m$ is the number of non-zero elements.
}
algorithm for \problem{Vertex-disjoint $s-t$ Paths}.
We obtain fully polynomial FPT algorithms for a wide range of problems by using this approach.
Our key observation is that, when using the divide-and-conquer approach, another graph parameter called \emph{tree-depth} is more powerful than the tree-width.

A graph $G$ of tree-width $\tw$ admits a set $S$ of $\tw+1$ vertices, called a balanced separator, such that each connected component of $G-S$ contains at most $\frac{n}{2}$ vertices.
In both of the above-mentioned divide-and-conquer algorithms for graphs of bounded tree-width by Akiba~et~al.~\cite{akiba2013fast} and Fomin~et~al.~\cite{Fomin17}, after the algorithm recursively computes a solution for each connected component of $G-S$,
it constructs a solution for $G$ in $O(\tw\cdot (m+n\log n))$ time or $O(\tw\cdot m)$ time, respectively.
Because the depth of the recursive calls is bounded by $O(\log n)$, the total running time becomes $O(\tw\cdot (m+n\log n)\log n)$ or $O(\tw\cdot m\log n)$, respectively.

Here, we observe that, by using tree-depth, this kind of divide-and-conquer algorithm can be simplified and the analysis can be improved.
Tree-depth is a graph parameter which has been studied under various names~\cite{Schaffer89,KatchalskiMS95,BodlaenderDJKKMT98,NesetrilM06}.
A graph has tree-depth $\td$ if and only if there exists an \emph{elimination forest} of depth $\td$.
See Section~\ref{sec:pre} for the precise definition of the tree-depth and the elimination forest.
An important property of tree-depth is that any connected graph $G$ of tree-depth $\td$ can be divided into connected components of tree-depth at most $\td-1$ by removing a single vertex $r$.
Therefore, if there exists an $O(m)$-time or $O(m+n\log n)$-time \emph{incremental algorithm}, which constructs a solution for $G$ from a solution for $G-r$,
we can solve the problem in $O(\td\cdot m)$ time or $O(\td\cdot (m+n\log n))$ time, respectively.
Now, the only thing to do is to develop such an incremental algorithm for each problem.
We present a detailed discussion of this framework in Section~\ref{sec:framework}.
Because any graph of tree-width $\tw$ has tree-depth at most $(\tw+1)\log_2 n$~\cite{DBLP:books/daglib/0030491},
the running time can also be bounded by $O(\tw\cdot m\log n)$ or $O(\tw\cdot (m+n\log n)\log n)$.
Therefore, our analysis using tree-depth is never worse than the existing results directly using tree-width.
On the other hand, there are infinitely many graphs whose tree-depth has asymptotically the same bound as tree-width.
For instance, if every $N$-vertex subgraph admits a balanced separator of size $O(N^\alpha)$ for some constant $\alpha>0$ (e.g., $\alpha=\frac{1}{2}$ for $H$-minor free graphs),
both tree-width and tree-depth are $O(n^\alpha)$.
Hence, for such graphs, the time complexity using tree-depth is truly better than that using tree-width.

\subparagraph{Our results.}

\begin{table}[t]
	\begin{center}
	\caption{Comparison of previous results and our results.
	$n$ and $m$ denotes the number of vertices and edges, $w$ denotes the width of the given tree-decomposition, and $d$ denotes the depth of the given elimination forest.
	The factor $d$ in our results can be replaced by $w\cdot \log n$.}\label{tab:summary}
  		\begin{tabular}{|l|l|l|} \hline
    Problem & Previous result & Our result\\ \hline \hline
    $\textsc{Maximum Matching}$ & $O(w^4 n\log^2 n)$~\cite{Fomin17} & $O(dm)$ \\ \hline
    $\textsc{Weighted Matching}$ & Open problem~\cite{Fomin17} & $O(d(m+n\log n))$ \\ \hline
    $\textsc{Negative Cycle Detection}$ & $O(w^2 n)$~\cite{planken2012computing} & $O(d(m+n\log n))$ \\ \hline
    $\textsc{Minimum Weight Cycle}$ & \multicolumn{1}{c|}{---} & $O(d(m+n\log n))$\\ \hline
    $\textsc{Replacement Paths}$ & \multicolumn{1}{c|}{---} & $O(d(m+ n\log n))$\\ \hline
    $\textsc{2-hop Cover}$ & $O(w(m+n\log n)\log n)$~\cite{akiba2013fast} & $O(d(m+n\log n))$\\ \hline
  		\end{tabular}
  	\end{center}
\end{table}

Table~\ref{tab:summary} shows our results and the comparison to the existing results on fully polynomial FPT algorithms parameterized by tree-width.
The formal definition of each problem is given in Section 4.
Because obtaining an elimination forest of the lowest depth is NP-hard, we assume that an elimination forest is given as an input
and the parameter for our results is the depth $d$ of the given elimination forest.
Similarly, for the existing results, the parameter is the width $w$ of the given tree-decomposition.
Note that, because a tree-decomposition of width $w$ can be converted into an elimination forest of depth $O(w\cdot \log n)$ in linear time~\cite{Schaffer89},
we can always replace the factor $d$ in our running time by $w\cdot \log n$.

The first polynomial-time algorithms for \problem{Maximum Matching} and \problem{Weighted Matching} were obtained by Edmonds~\cite{edmonds1965paths},
and the current fastest algorithms run in $O(\sqrt{n}m)$ time~\cite{blum1990new,vazirani1994theory,DBLP:journals/jacm/GabowT91}
and $O(n(m+n \log n))$ time~\cite{blum1990new}, respectively.
Fomin~et~al.~\cite{Fomin17}~obtained the $O(w^4 n\log^2 n)$-time randomized algorithm for \problem{Maximum Matching} by using an algebraic method and the fast computation of Gaussian elimination.
They left as an open problem whether a similar running time is possible for \problem{Weighted Matching}.
The general-case algorithms for these problems compute a maximum matching by iteratively finding an \emph{augmenting path}, and therefore, they are already incremental.
Thus, we can easily obtain an $O(dm)$-time algorithm for \problem{Maximum Matching} and an $O(d(m+n\log n))$-time algorithm for \problem{Weighted Matching}.
Note that the divide-and-conquer algorithms for planar matching by Lipton and Tarjan~\cite{LiptonT80} also uses this augmenting-path approach,
and our result can be seen as extension to bounded tree-depth graphs.
Our algorithm for \problem{Maximum Matching} is always faster\footnote{%
Note that for any graph of tree-width or tree-depth $k$, we have $m=O(kn)$.
} than the one by Fomin~et~al.~
and is faster than the general-case algorithm already when $d=O(n^{\frac{1}{2}-\epsilon})$.
Our algorithm for \problem{Weighted Matching} settles the open problem and is faster than the general-case algorithm already when $d=O(n^{1-\epsilon})$.

The current fastest algorithm for \problem{Negative Cycle Detection} is the classical $O(nm)$-time Bellman-Ford algorithm.
Planken~et~al.~\cite{planken2012computing} obtained an $O(w^2 n)$-time algorithm by using a Floyd-Warshall-like dynamic programming.
In this paper, we give an $O(d(m+n\log n))$-time algorithm.
While the algorithm by Planken~et~al.~is faster than the general-case algorithm only when $w=O(m^{\frac{1}{2}-\epsilon})$,
our algorithm achieves a faster running time already when $d=O(n^{1-\epsilon})$.

Both \problem{Minimum Weight Cycle} (or \problem{Girth}) and \problem{Replacement Paths} are subcubic-equivalent to \problem{APSP}~\cite{WilliamsW10}.
A naive algorithm can solve both problems in $O(n^3)$ time or $O(n(m+n\log n))$ time.
For \problem{Minimum Weight Cycle} of directed graphs, an improved $O(nm)$-time algorithm was recently obtained by Orlin and Sede\~{n}o-Noda~\cite{orlin2017onm}.
For \problem{Replacement Paths}, Malik~et~al.~\cite{malik1989k} obtained an $O(m + n \log n)$-time algorithm for undirected graphs,
and Roditty and Zwick~\cite{DBLP:journals/talg/RodittyZ12} obtained an $O(\sqrt{n} m \cdot \mathrm{polylog}\,n)$-time algorithm for unweighted graphs.
For the general case, Gotthilf and Lewenstein~\cite{DBLP:journals/ipl/GotthilfL09} obtained an $O(n(m+n\log\log n))$-time algorithm,
and there exists an $\Omega(\sqrt{n} m)$-time lower bound in the path-comparison model~\cite{karger1993finding} (whenever $m=O(n\sqrt{n})$)~\cite{hershberger2007difficulty}.
In this paper, we give an $O(d(m+n\log n))$-time algorithm for each of these problems, which is faster than the general-case algorithm already when $d=O(n^{1-\epsilon})$.
This result shows the following contrast to the known result of ``Hardness in P'':
\problem{Radius} is also subcubic-equivalent to \problem{APSP}~\cite{AbboudGW15} but it cannot be solved in a similar running time under some plausible assumptions~\cite{AbboudWW16}.


\emph{2-hop cover}~\cite{cohen2003reachability} is a data structure for efficiently answering distance queries.
Akiba~et~al.~\cite{akiba2013fast} obtained an $O(w(m+n\log n)\log n)$-time algorithm for constructing a 2-hop cover answering each distance query in $O(w\log n)$ time.
In this paper, we give an $O(d(m+n\log n))$-time algorithm for constructing a 2-hop cover answering each distance query in $O(d)$ time.

\subparagraph{Related work.}
Coudert, Ducoffe and Popa~\cite{Coudert17} have developed fully polynomial FPT algorithms using several other graph parameters including clique-width.
In contrast to the tree-depth, their parameters are not polynomially bounded by tree-width,
and therefore, their results do not imply fully polynomial FPT algorithms parameterized by tree-width.
Mertzios, Nichterlein and Niedermeier~\cite{DBLP:journals/corr/MertziosNN16} have obtained an $O(m+k^{1.5})$-time algorithm for \problem{Maximum Matching} parameterized by feedback edge number $k$ ($=m-n+1$ when the graph is connected) by giving a linear-time kernel.

\section{Preliminaries}\label{sec:pre}

Let $G=(V,E)$ be a directed or undirected graph, where $V$ is a set of vertices of $G$ and $E$ is a set of edges of $G$.
When the graph is clear from the context, we use $n$ to denote the number of vertices and $m$ to denote the number of edges.
All the graphs in this paper are simple (i.e., they have no multiple edges nor self-loops).
Let $S\subseteq V$ be a subset of vertices.
We denote by $E[S]$ the set of edges whose endpoints are both in $S$
and denote by $G[S]$ the subgraph induced by $S$ (i.e., $G[S]=(S,E[S])$).

A \emph{tree decomposition} of a graph $G=(V,E)$ is a pair $(T, B)$ of a tree $T=(X, F)$ and a collection of \emph{bags} $\{B_x\subseteq V\mid x\in X\}$ satisfying the following two conditions.
\begin{itemize}
\item For each edge $uv \in E$, there exists some $x\in X$ such that $\{u,v\}\subseteq B_x$.
\item For each vertex $v \in V$, the set $\{x\in X\mid v\in B_x\}$ induces a connected subtree in $T$.
\end{itemize}
The \emph{width} of $(T, B)$ is the maximum of $|B_x|-1$ and the \emph{tree-width} $\mathrm{tw}(G)$ of $G$ is the minimum width among all possible tree decompositions.

An \emph{elimination forest} $T$ of a graph $G=(V,E)$ is a rooted forest on the same vertex set $V$ such that, for every edge $uv\in E$, one of $u$ and $v$ is an ancestor of the other.
The \emph{depth} of $T$ is the maximum number of vertices on a path from a root to a leaf in $T$.
The \emph{tree-depth} $\mathrm{td}(G)$ of a graph $G$ is the minimum depth among all possible elimination forests.
Tree-width and tree-depth are strongly related as the following lemma shows.

\begin{lemma}[\cite{DBLP:books/daglib/0030491,Schaffer89}]\label{lem:pre:relation}
For any graph $G$, the following holds.
\[
\mathrm{tw}(G)+1\leq\mathrm{td}(G)\leq(\mathrm{tw}(G)+1)\log_2 n.
\]
Moreover, given a tree decomposition of width $k$, we can construct an elimination forest of depth $O(k\log n)$ in linear time.
\end{lemma}

\section{Divide-and-conquer framework}\label{sec:framework}
In this section, we propose a divide-and-conquer framework that can be applicable to a wide range of problems parameterized by tree-depth.

\begin{theorem}\label{thm:framework}
Let $G=(V,E)$ be a graph and let $f$ be a function defined on subsets of $V$.
Suppose that $f(\emptyset)$ can be computed in a constant time and we have the following two algorithms $\textsc{Increment}$ and $\textsc{Union}$
with time complexity $T(n,m)(=\Omega(n+m))$.
\begin{itemize}
  \item $\textsc{Increment}(X,f(X),x)\mapsto f(X\cup\{x\})$. Given a set $X\subseteq V$, its value $f(X)$, and a vertex $x\not\in X$,
  		this algorithm computes the value $f(X\cup\{x\})$ in $T(|X\cup\{x\}|,|E[X\cup\{x\}]|)$ time.
  \item $\textsc{Union}((X_1,f(X_1)),\ldots,(X_c,f(X_c)))\mapsto f(\bigcup_i X_i)$. Given disjoint sets $X_1,\ldots,X_c\subseteq V$ such that $G$ has no edges between $X_i$ and $X_j$ for any $i\neq j$,
  		and their values $f(X_1),\ldots,f(X_c)$,
  		this algorithm computes the value $f(\bigcup_i X_i)$ in $T(|\bigcup_i X_i|,|E[\bigcup_i X]|)$ time.
\end{itemize}
Then, for a given elimination forest of $G$ of depth $k$, we can compute the value $f(V)$ in $O(k\cdot T(n, m))$ time.
\end{theorem}

\begin{algorithm}[t!]
\caption{Algorithm for computing $f(V)$.}
\label{alg:framework}
\begin{algorithmic}[1]
\Procedure{Compute}{$S, T_S$} $\mapsto f(S)$\Comment{$T_S$ is an elimination forest of $G[S]$.}
\If{$S=\emptyset$}
	\Return $f(\emptyset)$
\EndIf
\State $T_1,\ldots,T_c\gets$ the connected trees of $T_S$
\State $X_1,\ldots,X_c\gets$ the sets of vertices of $T_1,\ldots,T_c$
\For{$i\in\{1,\ldots,c\}$}
	\State $x_i\gets$ the root of $T_i$
	\State $f_i\gets\textsc{Increment}(X_i\setminus\{x_i\}, \textsc{Compute}(X\setminus\{x_i\}, T_i-x_i), x_i)$\label{line:framework:rec}
\EndFor
\State\Return $\textsc{Union}((X_1,f_1),\ldots,(X_c,f_c))$
\EndProcedure
\end{algorithmic}
\end{algorithm}

\begin{proof}
Algorithm~\ref{alg:framework} describes our divide-and-conquer algorithm.
We prove that for any set $S$ and any elimination forest $T_S$ of $G[S]$ of depth $k_S$,
the procedure $\textsc{Compute}(S,T_S)$ correctly computes the value $f(S)$ in $(2k_S+1)\cdot T(|S|,|E[S]|)$ time by induction on the size of $S$.

The claim trivially holds when $S=\emptyset$.
For a set $S\neq\emptyset$, let $T_1,\ldots,T_c$ be the connected trees of $T_S$ ($c=1$ if $T_S$ is connected).
For each $i$, let $X_i$ be the set of vertices of $T_i$.
From the definition of the elimination forest, $G$ has no edges between $X_i$ and $X_j$ for any $i\neq j$.
For each $i$, we compute the value $f(X_i)$ as follows.
Let $x_i$ be the root of $T_i$.
By removing $x_i$ from $T_i$, we obtain an elimination forest of $G[X_i\setminus\{x_i\}]$ of depth at most $k_S-1$.
Therefore, by the induction hypothesis, we can correctly compute the value $f(X_i\setminus\{x_i\})$ in $(2k_S-1)\cdot T(|X_i|,|E[X_i]|)$ time.
Then, by applying $\textsc{Increment}(X_i\setminus\{x_i\}, f(X_i\setminus\{x_i\}), x_i)$, we obtain the value $f(X_i)$ in $2k_S\cdot T(|X_i|,|E[X_i]|)$ time.
Because $|S|=\sum_i |X_i|$ and $|E[S]|=\sum_i |E[X_i]|$ hold, the total running time of these computations is
$2k_S\cdot \sum_i T(|X_i|,|E[X_i]|)\leq 2k_S\cdot T(|S|,|E[S]|)$.
Finally, by applying the algorithm $\textsc{Union}$, we obtain the value $f(S)$ in $(2k_S+1)\cdot T(|S|,|E[S]|)$ time.
\end{proof}

Note that the algorithm \textsc{Union} is trivial in most applications.
We have only one non-trivial case in Section~\ref{sec:replacement} in this paper.
From the relation between tree-depth and tree-width (Lemma~\ref{lem:pre:relation}), we obtain the following corollary.
\begin{corollary}
Under the same assumption as in Theorem~\ref{thm:framework}, for a given tree decomposition of $G$ of width $k$, we can compute the value $f(V)$ in $O(k\cdot T(n, m)\log n)$ time.
\end{corollary}

\section{Applications}

\subsection{Maximum matching}
For an undirected graph $G=(V,E)$, a \emph{matching} $M$ of $G$ is a subset of $E$ such that no edges in $M$ share a vertex.
In this section, we prove the following theorem.
\begin{theorem}\label{thm:matching}
Given an undirected graph and its elimination forest of depth $k$, we can compute a maximum-size matching in $O(k m)$ time.
\end{theorem}

As mentioned in the introduction, we use the augmenting-path approach, which is also used for planar matching~\cite{LiptonT80}.
Let $M$ be a matching.
A vertex not incident to $M$ is called \emph{exposed}.
An \emph{$M$-alternating path} is a (simple) path whose edges are alternately out of and in $M$.
An $M$-alternating path connecting two different exposed vertices is called an \emph{$M$-augmenting path}.
If there exists an $M$-augmenting path $P$, by taking the symmetric difference $M\Delta E(P)$, where $E(P)$ is the set of edges in $P$, we can construct a matching of size $|M|+1$.
In fact, $M$ is the maximum-size matching if and only if there exist no $M$-augmenting paths.
Edmonds~\cite{edmonds1965paths} developed the first polynomial-time algorithm for computing an $M$-augmenting path by introducing the notion of blossom,
and an $O(m)$-time algorithm was given by Gabow and Tarjan~\cite{DBLP:journals/jcss/GabowT85}.
\begin{lemma}[\cite{DBLP:journals/jcss/GabowT85}]\label{lem:matching:augment}
Given an undirected graph and its matching $M$, we can either compute a matching of size $|M|+1$ or correctly conclude that $M$ is a maximum-size matching in $O(m)$ time.
\end{lemma}
For $S\subseteq V$, we define $f(S)$ as a function that returns a maximum-size matching of $G[S]$.
We now give algorithms \textsc{Increment} and \textsc{Union}.

\subsubsection*{$\textsc{Increment}(X,f(X),x)$.}
Because the size of the maximum matching of $G[X\cup\{x\}]$ is at most the size of the maximum matching of $G[X]$ plus one,
we can compute a maximum matching of $G[X\cup\{x\}]$ in $O(|E[X\cup\{x\}]|)$ time by a single application of Lemma~\ref{lem:matching:augment}.

\subsubsection*{$\textsc{Union}((X_1,f(X_1)),\ldots,(X_c,f(X_c)))$.}
Because there exist no edges between $X_i$ and $X_j$ for any $i\neq j$, we can construct a maximum matching of $G[\bigcup_i X_i]$ just by taking the union of $f(X_i)$.

\begin{proof}[Proof of Theorem~\ref{thm:matching}]
The algorithm $\textsc{Increment}(X,f(X),x)$ correctly computes $f(X\cup\{x\})$ in $O(|E[X\cup\{x\}]|)$ time and
the algorithm $\textsc{Union}((X_1,f(X_1)),\ldots,(X_c,f(X_c)))$ correctly computes $f(\bigcup_i X_i)$ in $O(|\bigcup_i X_i|)$ time.
Therefore, from Theorem~\ref{thm:framework}, we can compute a maximum-size matching of $G$ in $O(k m)$ time.
\end{proof}

\subsection{Weighted matching}
Let $G=(V,E)$ be an undirected graph with an edge-weight function $w:E\to\R$.
A weight of a matching $M$, denoted by $w(M)$, is simply defined as the total weight of edges in $M$.
A matching $M$ of $G$ is called \emph{perfect} if $G$ has no exposed vertices (or equivalently $|M|=\frac{n}{2}$).
A perfect matching is called a \emph{maximum-weight perfect matching} if it has the maximum weight among all perfect matchings of $G$.
We can easily see that other variants of weighted matching problems can be reduced to the problem of finding a maximum-weight perfect matching
even when parameterized by tree-depth (see Appendix~\ref{sec:reductions}).
In this section, we prove the following theorem.
\begin{theorem}\label{thm:wmatching}
Given an edge-weighted undirected graph admitting at least one perfect matching and its elimination forest of depth $k$, we can compute a maximum-weight perfect matching in $O(k (m+n\log n))$ time.
\end{theorem}

In our algorithm, we use an $O(n(m+n\log n))$-time primal-dual algorithm by Gabow~\cite{DBLP:conf/soda/Gabow90}.
In this primal-dual algorithm, we keep a pair of a matching $M$ and dual variables $(\Omega, y, z)$, where $\Omega$ is a laminar\footnote{%
A collection $\Omega$ of subsets of a ground set $V$ is called \emph{laminar} if for any $X,Y\in\Omega$, one of $X\cap Y=\emptyset$, $X\subseteq Y$, or $X\subseteq Y$ holds.
When $\Omega$ is laminar, we have $|\Omega|=O(|V|)$.
} collection of odd-size subsets of $V$
and $y$ and $z$ are functions $y:V\to\R$ and $z:\Omega\to\R_{\geq 0}$, satisfying the following conditions:
\begin{eqnarray}
\widehat{yz}(uv):=y(u)+y(v)+\sum_{B\in\Omega:u,v\in B}z(B)\geq w(uv)&&\text{for every } uv\in E,\label{cond:wmatching:1}\\
\widehat{yz}(uv)=w(uv)&&\text{for every } uv\in M,\label{cond:wmatching:2}\\
|\{uv\in M\mid u,v\in B\}|=\left\lfloor\frac{|B|}{2}\right\rfloor&&\text{for every }B\in\Omega.\label{cond:wmatching:3}
\end{eqnarray}
From the duality theory (see e.g.~\cite{DBLP:journals/corr/Gabow16a}), a perfect matching $M$ is a maximum-weight perfect matching if and only if there exist dual variables $(\Omega, y, z)$ satisfying the above conditions.
Gabow~\cite{DBLP:conf/soda/Gabow90} obtained the $O(n(m+n\log n))$-time algorithm by iteratively applying the following lemma.

\begin{lemma}[\cite{DBLP:conf/soda/Gabow90}]\label{lem:wmatching:augment}
Given an edge-weighted undirected graph and a pair of a matching $M$ and dual variables $(\Omega, y, z)$ satisfying the conditions (\ref{cond:wmatching:1})--(\ref{cond:wmatching:3}),
we can either compute a pair of a matching $M'$ of cardinality $|M|+1$ and dual variables $(\Omega', y', z')$ satisfying the conditions (\ref{cond:wmatching:1})--(\ref{cond:wmatching:3})
or correctly conclude that $M$ is a maximum-size matching\footnote{%
Note that when $M$ is not a perfect matching, this does not imply that $M$ has the maximum weight among all the maximum-size matchings.
}
in $O(m+n\log n)$ time.
\end{lemma}

For $S\subseteq V$, we define $f(S)$ as a function that returns a pair of a maximum-size matching $M_S$ of $G[S]$ and dual variables $(\Omega_S, y_S, z_S)$
satisfying the conditions (\ref{cond:wmatching:1})--(\ref{cond:wmatching:3}).
We now give algorithms \textsc{Increment} and \textsc{Union}.

\subsubsection*{$\textsc{Increment}(X,f(X),x)$.}
Let $W$ be a value satisfying $W+y_X(v)\geq w(xv)$ for every $xv\in E[X\cup\{x\}]$.
Let $y:X\cup\{x\}\to\R$ be a function defined as $y(x):=W$ and $y(v):=y_X(v)$ for $v\in X$.
In the subgraph $G[X\cup \{x\}]$, a pair of the matching $M_X$ and dual variables $(\Omega_X, y, z_X)$ satisfies the conditions (\ref{cond:wmatching:1})--(\ref{cond:wmatching:3}).
Therefore, we can apply Lemma~\ref{lem:wmatching:augment}.
If $M_X$ is a maximum-size matching of $G[X\cup \{x\}]$, we return $M_X$ and $(\Omega_X, y, z_X)$.
Otherwise, we obtain a matching $M'$ of size $|M_X|+1$ and dual variables $(\Omega', y', z')$ satisfying the conditions (\ref{cond:wmatching:1})--(\ref{cond:wmatching:3}).
Because the cardinality of maximum-size matching of $G[X\cup\{x\}]$ is at most the cardinality of maximum-size matching of $G[X]$ plus one, the obtained $M'$ is a maximum-size matching of $G[X\cup\{x\}]$.
Therefore, we can return $M'$ and $(\Omega', y', z')$.

\subsubsection*{$\textsc{Union}((X_1,f(X_1)),\ldots,(X_c,f(X_c)))$.}
Because there exist no edges between $X_i$ and $X_j$ for any $i\neq j$,
we can simply return a pair of a maximum-size matching obtained by taking the union $\bigcup_i M_{X_i}$ and dual variables $(\Omega, y, z)$ such that
$\Omega:=\bigcup_i \Omega_{X_i}$, $y(v):=y_{X_i}(v)$ for $v\in X_i$, and $z(B)=z_{X_i}(B)$ for $B\in\Omega_{X_i}$.

\begin{proof}[Proof of Theorem~\ref{thm:wmatching}]
The algorithm $\textsc{Increment}(X,f(X),x)$ runs in $O(|E[X\cup\{x\}]|+|X|\log |X|)$ time and
the algorithm $\textsc{Union}((X_1,f(X_1)),\ldots,(X_c,f(X_c)))$ runs in $O(|\bigcup X_i|)$ time.
Therefore, from Theorem~\ref{thm:framework}, we can compute $f(V)$ in $O(k (m+n\log n))$ time.
From the duality theory, the perfect matching obtained by computing $f(V)$ is a maximum-weight perfect matching of $G$.
\end{proof}

\subsection{Negative cycle detection and potentials}
Let $G=(V,E)$ be a directed graph with an edge-weight function $w:E\to\R$.
For a function $p:V\to\R$, we define an edge-weight function $w_p$ as $w_p(uv):=w(uv)+p(u)-p(v)$.
If $w_p$ becomes non-negative for all edges, $p$ is called a \emph{potential on $G$}.

\begin{lemma}[\cite{schrijver2003}]\label{potential}
There exists a potential on $G$ if and only if $G$ has no negative cycles.
\end{lemma}

In this section, we prove the following theorem.
\begin{theorem}\label{thm:negative_cycle}
Given an edge-weighted directed graph and its elimination forest of depth $k$, we can compute either a potential or a negative cycle in $O(k (m+n\log n))$ time.
\end{theorem}
Suppose that we have a potential $p$.
Because $w_p$ is non-negative, we can compute a shortest-path tree rooted at a given vertex $s$ under $w_p$ in $O(m + n\log n)$ time by Dijkstra's algorithm.
For any $s-t$ path, its length under $w_p$ is exactly the length under $w$ plus a constant $p(s)-p(t)$.
Therefore, the obtained tree is also a shortest-path tree under $w$.
Thus, we obtain the following corollary.
\begin{corollary}
Given an edge-weighted directed graph without negative cycles, a vertex $s$, and its elimination forest of depth $k$,
we can compute a shortest-path tree rooted at $s$ in $O(k (m+n\log n))$ time.
\end{corollary}

For $S\subseteq V$, we define $f(S)$ as a function that returns either a potential $p_S:S\to\R$ on $G[S]$ or a negative cycle contained in $G[S]$.
We now give algorithms \textsc{Increment} and \textsc{Union}.

\subsubsection*{$\textsc{Increment}(X,f(X),x)$.}
If $f(X)$ is a negative cycle, we return it.
Otherwise, let $G'=(X\cup\{x\},E')$ be the graph obtained from $G[X\cup\{x\}]$ by removing all the edges incoming to $x$.
Let $W$ be a value satisfying $w(xv)+W-p_X(v)\geq 0$ for every $xv\in E'$.
Let $p':X\cup\{x\}\to\R$ be a function defined as $p'(x):=W$ and $p'(v):=p_X(v)$ for $v\in X$.
Because $x$ has no incoming edges in $G'$, $p'$ is a potential on $G'$.
Therefore, we can compute a shortest-path tree rooted at $x$ under $w_{p'}$ in $O(|E[X]|+|X|\log |X|)$ time by Dijkstra's algorithm.
Let $R$ be the set of vertices reachable from $x$ in $G'$ and let $d:R\to \R$ be the shortest-path distance from $x$ under $w_{p'}$.
If there exists an edge $vx\in E[X\cup\{x\}]$ such that $v\in R$ and $d(v)+w_{p'}(vx)<0$, $G[X\cup\{x\}]$ contains a negative cycle
starting from $x$, going to $v$ along the shortest-path tree, and coming back to $x$ via the edge $vx$.
Otherwise, let $D$ be a value satisfying $w_{p'}(uv)+D-d(v)\geq 0$ for every $uv\in E[X\cup\{x\}]$ with $u\in X\setminus R$ and $v\in R$.
Then, we return a function $p:X\cup\{x\}\to\R$ defined as $p(v):=p'(v)+d(v)$ if $v\in R$ and $p(v):=p'(v)+D$ if $v\in X\setminus R$.
\begin{claim}
$p$ is a potential on $G[X\cup\{x\}]$.
\end{claim}
\begin{proof}
For every edge $uv\in E[X\cup\{x\}]$, we have
\[
w_p(uv)=\begin{cases}
w_{p'}(uv)+d(u)-d(v)\geq 0&\text{if } u,v\in R,\\
w_{p'}(uv)+D-d(v)\geq 0&\text{if }u\in X\setminus R, v\in R,\\
w_{p'}(uv)+D-D\geq 0&\text{if }u\in X\setminus R, v\in X\setminus R.
\end{cases}
\]
Note that there are no edges from $R$ to $X\setminus R$.
\end{proof}

\subsubsection*{$\textsc{Union}((X_1,f(X_1)),\ldots,(X_c,f(X_c)))$.}
If at least one of $f(X_i)$ is a negative cycle, we return it.
Otherwise, we return a potential $p$ defined as $p(v):=p_{X_i}(v)$ for $v\in X_i$.

\begin{proof}[Proof of Theorem~\ref{thm:negative_cycle}]
The algorithm $\textsc{Increment}(X,f(X),x)$ correctly computes $f(X\cup\{x\})$ in $O(|E[X]|+|X|\log |X|)$ time and
the algorithm $\textsc{Union}((X_1,f(X_1)),\ldots,(X_c,f(X_c)))$ correctly computes $f(\bigcup_i X_i)$ in $O(|\bigcup_i X_i|)$ time.
Therefore, from Theorem~\ref{thm:framework}, we can compute $f(V)$, i.e., either a potential on $G$ or a negative cycle contained in $G$, in $O(k(m+n\log n))$ time.
\end{proof}

\subsection{Minimum weight cycle}
In this section, we prove the following theorem.

\begin{theorem}\label{minimum weight cycle}
Given a non-negative edge-weighted undirected or directed graph and its elimination forest of depth $k$,
we can compute a minimum-weight cycle in $O(k (m + n\log n))$ time.
\end{theorem}

Note that when the graph is undirected, a closed walk of length two using the same edge twice is not considered as a cycle.
Therefore, we cannot simply reduce the undirected version into the directed version by replacing each undirected edge by two directed edges of both directions.


Let $G=(V,E)$ be the input graph with an edge-weight function $w:E\to\R_{\geq 0}$.
For $S\subseteq V$, we define $f(S)$ as a function that returns a minimum-weight cycle of $G[S]$.
We describe \textsc{Increment} and \textsc{Union} below.

\subsubsection*{$\textsc{Increment}(X,f(X),x)$.}
Because we already have a minimum-weight cycle $f(X)$ of $G[X]$, we only need to find a minimum-weight cycle passing through $x$.
First, we construct a shortest-path tree of $G[X\cup\{x\}]$ rooted at $x$ and let $d: X\cup\{x\}\to\R$ be the shortest-path distance.

When the graph is undirected,
we find an edge $uv\in E[X\cup\{x\}]$ not contained in the shortest-path tree minimizing $d(u) + w(uv) + d(v)$.
If this weight is at least the weight of $f(X)$, we return $f(X)$.
Otherwise, we return the cycle starting from $x$, going to $u$ along the shortest-path tree, jumping to $v$ through the edge $uv$, and coming back to $x$ along the shortest-path tree.
Note that this always forms a cycle because otherwise, it induces a cycle contained in $G[X]$ that has a smaller weight than $f(X)$, which is a contradiction.

We can prove the correctness of this algorithm as follows.
Let $W$ be the weight of the cycle obtained by the algorithm and let $C$ be a cycle passing through $x$.
Let $v_0=x,v_1,\ldots,v_{\ell-1},v_\ell=x$ the vertices on $C$ in order.
Because a tree contains no cycles, there exists an edge $v_iv_{i+1}$ not contained in the shortest-path tree.
Therefore, the weight of $C$ is
$\sum_{j=0}^{i-1}w(v_jv_{j+1})+w(v_iv_{i+1})+\sum_{j=i+1}^{\ell-1}w(v_jv_{j+1})\geq d(v_i)+w(v_iv_{i+1})+d(v_{i+1})\geq W$.

When the graph is directed,
we find an edge $ux\in E[X\cup\{x\}]$ with the minimum $d(u)+w(ux)$.
If this weight is at least the weight of $f(X)$, we return $f(X)$.
Otherwise, we return the cycle starting from $x$, going to $u$ along the shortest-path tree, and coming back to $x$ through the edge $ux$.

\subsubsection*{$\textsc{Union}((X_1,f(X_1)),\ldots,(X_c,f(X_c)))$.}
We return a cycle of the minimum weight among $f(X_1), \ldots, f(X_c)$.

\begin{proof}[Proof of Theorem~\ref{minimum weight cycle}]
The algorithm $\textsc{Increment}(X,f(X),x)$ correctly computes $f(X\cup\{x\})$ in $O(|E[X]|+|X|\log |X|)$ time and
the algorithm $\textsc{Union}((X_1,f(X_1)),\ldots,(X_c,f(X_c))$ correctly computes $f(\bigcup_i X_i)$ in $O(|\bigcup_i X_i|)$ time.
Therefore, from Theorem~\ref{thm:framework}, we can compute a minimum-weight cycle in $O(k(m+n\log n))$ time.
\end{proof}

\subsection{Replacement paths}\label{sec:replacement}
Fix two vertices $s$ and $t$.
For an edge-weighed directed graph $G=(V,E)$ and an edge $e\in E$, we denote the length of the shortest $s-t$ path avoiding $e$ by $r_G(e)$.
In this section, we prove the following theorem.

\begin{theorem}\label{replacement paths}
Given an edge-weighted directed graph $G=(V,E)$, a shortest $s-t$ path $P$, and its elimination forest of depth $k$,
we can compute $r_G(e)$ for all edges $e$ on $P$ in $O(k(m+n\log n))$ time.
\end{theorem}

Let $v_0(=s),v_1,\ldots,v_{\ell-1},v_\ell(=t)$ be the vertices on the given shortest $s-t$ path $P$ in order.
For $i\in\{0,\ldots,\ell\}$, we denote the length of the prefix $v_0v_1\ldots v_i$ by $\mathrm{pref}(v_i)$ and the length of the suffix $v_iv_{i+1}\ldots v_\ell$ by $\mathrm{suf}(v_i)$.
These can be precomputed in linear time.

For $S\subseteq V$, we define $G[S]\cup P$ as a graph consisting of vertices $S\cup\{v_0,\ldots,v_\ell\}$ and edges $E[S]\cup\{v_0v_1,\ldots,v_{\ell-1}v_\ell\}$,
and define $G[S]\setminus P$ as a graph consisting of vertices $S$ and edges $E[S]\setminus\{v_0v_1,\ldots,v_{\ell-1}v_\ell\}$.
We denote the shortest-path length from $u$ to $v$ in $G[S]\setminus P$ by $d_S(u,v)$.
For convenience, we define $d_S(u,v)=\infty$ when $u\not\in S$ or $v\not\in S$.
We use the following lemma.
\begin{lemma}\label{lem:replacement:structure}
For any $S\subseteq V$ and any $i\in\{0,\ldots,\ell-1\}$, $r_{G[S]\cup P}(v_i v_{i+1})$ is the minimum of
$\mathrm{pref}(v_a)+d_S(v_a,v_b)+\mathrm{suf}(v_b)$ for $a\leq i<b$.
\end{lemma}
\begin{proof}
Any $s-t$ path avoiding $v_iv_{i+1}$ in $G[S]\cup P$ can be written as, for some $a\leq i<b$, a concatenation of
$s-v_a$ path $Q_1$, $v_a-v_b$ path $Q_2$ that is contained in $G[S]\setminus P$, and $v_b-t$ path $Q_3$.
Because $P$ is a shortest $s-t$ path in $G$, we can replace
$Q_1$ by the prefix $v_0\ldots v_a$, $Q_2$ by the shortest $v_a-v_b$ path in $G[S]\setminus P$, and $Q_3$ by the suffix $v_b\ldots v_\ell$ without increasing the length.
Therefore, the lemma holds.
\end{proof}

We want to define $f(S)$ as a function that returns a list of $r_{G[S]\cup P}(v_i v_{i+1})$ for all $i\in\{0,\ldots,\ell-1\}$;
however, we cannot do so because the length of this list is not bounded by $|S|$.
Instead, we define $f(S)$ as a function that returns a list of $r_{G[S]\cup P}(v_iv_{i+1})$ for all $i$ with $v_i\in S$.
This succinct representation has enough information because, for any $v_i\not\in S$, we have $r_{G[S]\cup P}(v_iv_{i+1})=r_{G[S]\cup P}(v_{i-1}v_i)$ (or $\infty$ when $i=0$).
We describe \textsc{Increment} and \textsc{Union} below.

\subsubsection*{$\textsc{Increment}(X,f(X),x)$.}
By running Dijkstra's algorithm twice, we can compute $d_{X\cup\{x\}}(x, v)$ and $d_{X\cup\{x\}}(v, x)$ for all $v\in X\cup\{x\}$ in $O(|E[X]|+|X|\log |X|)$ time.
For $v_i\in X\cup\{x\}$, 
we define $L_i:=\min_{a\leq i,v_a\in X\cup\{x\}}(\mathrm{pref}(v_a)+d(v_a,x))$ and
$R_i:=\min_{b>i,v_b\in X\cup\{x\}}(d(x,v_b)+\mathrm{suf}(v_b))$.
By a standard dynamic programming, we can compute $L_i$ and $R_i$ for all $i$ with $v_i\in X\cup\{x\}$ in $O(|X|)$ time.

From Lemma~\ref{lem:replacement:structure}, $r_{G[X\cup\{x\}]\cup P}(v_i v_{i+1})=\mathrm{pref}(v_a)+d_{X\cup\{x\}}(v_a,v_b)+\mathrm{suf}(v_b)$ holds for some $a\leq i<b$.
If $d_{X\cup\{x\}}(v_a,v_b)=d_X(v_a,v_b)$ holds, we have $r_{G[X\cup\{x\}]\cup P}(v_i v_{i+1})=r_{G[X]\cup P}(v_i v_{i+1})$,
and otherwise, we have $d_{X\cup\{x\}}(v_a,v_b)=d_{X\cup\{x\}}(v_a,x)+d_{X\cup\{x\}}(x,v_b)$.
Therefore, we can compute $r_{G[X\cup\{x\}]\cup P}(v_i v_{i+1})$ by taking the minimum of $r_{G[X]\cup P}(v_i v_{i+1})$ and
$\min_{a\leq i<b}(\mathrm{pref}(v_a)+d(v_a,x)+d(x,v_b)+\mathrm{suf}(v_b))=L_i+R_i$.

\subsubsection*{$\textsc{Union}((X_1,f(X_1)),\ldots,(X_c,f(X_c)))$.}
Let $X:=\bigcup_i X_i$.
Because there exist no edges between $X_i$ and $X_j$ for any $i\neq j$, we have $d_X(u,v)=\min_i d_{X_i}(u,v)$ for any $u,v\in X$.
Therefore, from Lemma~\ref{lem:replacement:structure}, we have $r_{G[X]\cup P}(v_i v_{i+1})=\min_j r_{G[X_j]\cup P}(v_i v_{i+1})$.
For efficiently computing $r_{G[X]\cup P}(v_i v_{i+1})$ for all $i$ with $v_i\in X$, we do as follows in increasing order of $i$.

For each $X_j$, we maintain a value $r_j$ so that $r_j=r_{G[X_j]\cup P}(v_i v_{i+1})$ always holds.
Initially, these values are set to $\infty$.
We use a heap for computing $\min_j r_j$ and updating $r_j$ in $O(\log c)$ time.
For processing $i$, we first update $r_j\gets r_{G[X_j]\cup P}(v_i v_{i+1})$ for the set $X_j$ containing $v_i$.
We do not need to update $r_{j'}$ for any other set $X_{j'}$ because $r_{G[X_{j'}]\cup P}(v_i v_{i+1})=r_{G[X_{j'}]\cup P}(v_{i-1} v_i)$ holds.
Then, we compute $r_{G[X]\cup P}(v_i v_{i+1})=\min_j r_j$.

\begin{proof}[Proof of Theorem~\ref{replacement paths}]
The algorithm $\textsc{Increment}(X,f(X),x)$ correctly computes $f(X\cup\{x\})$ in $O(|E[X]|+|X|\log |X|)$ time and
the algorithm $\textsc{Union}((X_1,f(X_1)),\ldots,(X_c,f(X_c)))$ correctly computes $f(\bigcup_i X_i)$ in $O(|\bigcup_i X_i|\log c)=O(|\bigcup_i X_i|\log |\bigcup_i X_i|)$ time.
Therefore, from Theorem~\ref{thm:framework}, we can compute $f(V)$, i.e., $r_{G\cup P}(e)=r_G(e)$ for all edges $e$ on $P$, in $O(k(m+n\log n))$ time.
\end{proof}

\subsection{2-hop cover}
Let $G=(V, E)$ be a directed graph with an edge-weight function $w: E\to\R_{\geq 0}$.
A \emph{2-hop cover} of $G$ is the following data structure $(L^+,L^-)$ for efficiently answering distance queries.
For each vertex $u\in V$, we assign a set $L^+(u)$ of pairs $(v,d^+_{uv})\in V\times\R_{\geq 0}$
and a set $L^-(u)$ of pairs $(v,d^-_{vu})\in V\times \R_{\geq 0}$.
We require that, for every pair of vertices $s,t\in V$, the shortest-path distance from $s$ to $t$ is exactly the minimum of
$d^+_{sh} + d^-_{ht}$ among all pairs $(h,d^+_{sh})\in L^+(s)$ and $(h,d^-_{ht})\in L^-(t)$.
The \emph{size} of the 2-hop cover is defined as $\sum_{u\in V}|L^+(u)|+|L^-(u)|$, and the \emph{maximum label size} is defined as $\max_{u\in V}|L^+(u)|+|L^-(u)|$.
Using a 2-hop cover of maximum label size $T$, we can answer a distance query in $O(T)$ time.
In this section, we prove the following theorem.

\begin{theorem}\label{thm:2-hop}
Given a non-negative edge-weighted directed graph and its elimination forest of depth $k$, we can construct a 2-hop cover of maximum label size $2k$ in $O(k(m+n\log n))$ time.
\end{theorem}

For $S\subseteq V$, we define $f(S)$ as a function that returns a 2-hop cover of $G[S]$.
We denote the shortest-path distance from $s$ to $t$ in $G[S]$ by $d_S(s,t)$.
We denote the result of the distance query from $s$ to $t$ for $f(S)$ by $q_S(s,t)$.
We now describe algorithms \textsc{Increment} and \textsc{Union}.

\subsubsection*{$\textsc{Increment}(X,f(X),x)$.}
Let $(L^+,L^-)$ be the 2-hop cover of $G[X]$.
By running Dijkstra's algorithm twice, we compute the shortest-path distances from $x$ and to $x$ in $G[X\cup\{x\}]$.
Then, for each $u\in X\cup\{x\}$, we insert $(x,d_{X\cup\{x\}}(u,x))$ into $L^+(u)$ and $(x,d_{X\cup\{x\}}(x,u))$ into $L^-(u)$.
Finally, we return the updated $(L^+,L^-)$ as $f(X\cup\{x\})$.

\begin{claim}
$f(X\cup\{x\})$ is a 2-hop cover of $G[X\cup\{x\}]$.
\end{claim}
\begin{proof}
It suffices to show that $q_{X\cup\{x\}}(s,t)=d_{X\cup\{x\}}(s,t)$ holds for every $s,t\in X\cup\{x\}$.
The claim clearly holds when $s=x$ or $t=x$.
For $s,t\in X$, let $\delta:=d_{X\cup\{x\}}(s,x)+d_{X\cup\{x\}}(x,t)$.
Then, we have $d_{X\cup\{x\}}(s,t)=\min(d_X(s,t),\delta)$.
From the construction of $f(X\cup\{x\})$, we have $q_{X\cup\{x\}}=\min(q_X(s,t),\delta)=\min(d_X(s,t),\delta)$.
Therefore, the claim holds.
\end{proof}

\subsubsection*{$\textsc{Union}((X_1,f(X_1)),\ldots,(X_c,f(X_c)))$.}
Because there exist no paths connecting $X_i$ and $X_j$ for any $i\neq j$, we can construct a 2-hop cover of $G[\bigcup_i X_i]$ by
simply concatenating the 2-hop covers $f(X_1), \ldots, f(X_c)$.

\begin{proof}[Proof of Theorem~\ref{thm:2-hop}]
The algorithm $\textsc{Increment}(X,f(X),x)$ correctly computes $f(X\cup\{x\})$ in $O(|E[X]| + |X|\log |X|)$ time and
the algorithm $\textsc{Union}((X_1,f(X_1)),\ldots,(X_c,f(X_c))$ correctly computes $f(\bigcup_i X_i)$ in $O(|\bigcup_i X_i|)$ time.
Therefore, from Theorem~\ref{thm:framework}, we can compute a 2-hop cover in $O(k(m + n\log n))$ time.
Let $(L^+,L^-)$ be the 2-hop cover obtained by computing $f(V)$.
For each element $(u,d^+_{uv})\in L^+(u)$ or $(u,d^-_{vu})\in L^-(u)$, $v$ is located on the path from $u$ to the root in the elimination forest.
Therefore, we have $|L^+(u)|+|L^-(u)|\leq 2k$ for every vertex $u\in V$.
\end{proof}

\section*{Acknowledgement}
We would like to thank Hiroshi Imai for pointing out to us the reference~\cite{LiptonT80}.

\bibliography{paper}

\appendix
\section{Reductions to weighted matching}\label{sec:reductions}
We present standard reductions from several problems to \problem{Maximum-weight Perfect Matching}.
Given a graph of $n$ vertices, $m$ edges, and tree-depth $k$, each of these reductions constructs a graph of $O(n)$ vertices, $O(m)$ edges, and tree-depth $O(k)$ in linear time.
Therefore, the problem can be solved in $O(k(m+n\log n))$ time.

First, we present a reduction from \problem{Maximum-weight Matching}.
Given an undirected graph $G=(V, E)$ with an edge-weight function $w:E\to\R$, our task is to find a maximum-weight matching of $G$ (which does not need to have the maximum cardinality).
We create a copy $G'=(V',E')$ of $G$ and connect each vertex $v\in V$ and its copy $v'\in V'$ by an edge of weight $0$.
Let $M$ be a maximum-weight matching of the original graph.
We can construct a perfect matching $M'$ of the reduced graph of weight $2w(M)$ by taking $M$, the copy of $M$ in $G'$, and the edge $vv'$ for every exposed vertex $v$.
Any perfect matching of the reduced graph consists of a matching of $G$, a matching of $G'$, and edges of weight zero,
and therefore, it has weight at most $2w(M)$.
Thus, by computing a maximum-weight perfect matching of the reduced graph, we can compute a maximum-weight matching of the original graph.
The reduced graph has $2n$ vertices and $2m+n$ edges.
From a given elimination forest of depth $k$ for the original graph,
we can construct an elimination forest of depth $2k$ for the reduced graph by replacing each vertex $v$ in the forest by a path $vv'$.

Next, we present a reduction from \problem{Maximum-weight Maximum-size Matching}, which is a problem of finding a matching of maximum weight subject to the constraint that it has the maximum cardinality.
We only modify the edge-weight function without modifying the graph.
From the edge-weight function $w:E\to\R$ of the input graph, we construct an edge-weight function $w'$ such that $w'(e):=w(e)+W$ for a large value $W$ satisfying $W>\sum_{e\in E} |w(e)|$.
For any matchings $M$ and $M'$ with $|M|>|M'|$, we have $w'(M)-w'(M')\geq W-\sum_{e\in E}|w(e)|>0$,
and for any matchings $M$ and $M'$ of the same cardinality, we have $w'(M)-w'(M')=w(M)-w(M')$.
Therefore, $M$ is a maximum-weight maximum-size matching under $w$ if and only if $M$ is a maximum-weight matching under $w'$.

Finally, we present a reduction from \problem{Minimum-weight Disjoint $A$-Paths}.
For a set of terminals $A\subseteq V$, a path is called an \emph{$A$-path} if it connects two distinct vertices in $A$ and has no internal vertex in $A$.
A set of $A$-paths is called \emph{disjoint} if no two paths share a vertex (including the terminals $A$).
Given a directed graph $G=(V, E)$ with an edge-weight function $w:E\to\R_{\geq 0}$ and a set of terminals $A\subseteq V$,
our task is to find disjoint $A$-paths of the minimum total weight subject to the constraint that the number of paths is maximized.
Note that we can easily reduce the problem of finding a minimum-weight vertex-disjoint $S-T$ paths to this problem
by removing all the edges incoming to $S$ or outgoing from $T$ and setting $A=S\cup T$.

The following reduction to \problem{Minimum-weight Maximum-size Matching} is based on the reduction for the unweighted version by Kriesell~\cite{DBLP:journals/jct/Kriesell05}.
We construct a graph $G'$ as follows.
For each vertex $v\in V\setminus A$, we create two vertices $\{v^+, v^-\}$ and an edge $v^+v^-$ of weight $0$.
For each terminal $a\in A$, we create a single vertex $a=a^+=a^-$.
For each edge $uv\in E$, we insert an edge $u^+ v^-$ of weight $w(uv)$.
$G'$ has at most $2n$ vertices and $2m+n$ edges.
From a given elimination forest of depth $k$ for $G$,
we can construct an elimination forest of depth $2k$ for $G'$ by replacing each vertex $v\in V\setminus A$ in the forest by a path $v^+v^-$.
Finally, we prove the correctness of the reduction.

From a set $\mathcal{P}$ of $d$ disjoint $A$-paths in $G$, we can construct a matching of cardinality $|V\setminus A|+d$ of the same weight in $G'$
by taking $u^+v^-$ for each $uv\in E$ used in $\mathcal{P}$ and $v^+v^-$ for each $v\in V\setminus A$ not used in $\mathcal{P}$.
From a matching $M$ of cardinality $|V\setminus A|+d$ in $G'$, we can construct a set $\mathcal{P}$ of $d$ disjoint $A$-paths of the same or smaller weight in $G$ as follows.
If one of $v^+$ and $v^-$ is exposed for some $v\in V\setminus A$, we obtain another matching without increasing the weight by
discarding the edge in $M$ incident to $v^+$ or $v^-$ and by including the edge $v^+v^-$.
Now we can assume that every exposed vertex is in $A$, and therefore the set of edges $\{uv\mid u^+v^-\in M, u\neq v\}$ induces a set $\mathcal{P}$ of $d$ disjoint $A$-paths and some cycles.
Because $w$ is non-negative, we can discard the cycles without increasing the weight.
Therefore, from the minimum-weight maximum-size matching of $G'$, we can construct a minimum-weight disjoint $A$-paths of $G$ in linear time.

\end{document}